\documentclass[a4paper,conference]{IEEEtran}
\IEEEoverridecommandlockouts

\usepackage{amsthm}
\usepackage{amssymb}
\usepackage{amsmath}
\usepackage{amssymb}
\usepackage{epsfig}
\usepackage{epsf}
\usepackage{subfigure}
\usepackage{graphicx}
\usepackage{cite}
\usepackage{url}
\usepackage[]{authblk}

\def\BibTeX{{\rm B\kern-.05em{\sc i\kern-.025em b}\kern-.08em
    T\kern-.1667em\lower.7ex\hbox{E}\kern-.125emX}}

\newtheorem{theorem}{Theorem}

\newtheorem{lemma}{Lemma}

\begin{document}

\title{When Does Spatial Correlation Add Value to Delayed Channel State Information?}

\author{
Alireza~Vahid\\
\and
Robert~Calderbank
  \thanks{Alireza~Vahid and Robert~Calderbank are with the Department of ECE, Duke University. Emails: {\sffamily alireza.vahid@duke.edu}, {\sffamily robert.calderbank@duke.edu}.}
  \thanks{The work of A. Vahid and R. Calderbank was funded in part by the Air Force Office of Scientific Research under grant AFOSR FA 9550-13-1-0076.}
}

\maketitle


\begin{abstract}
Fast fading wireless networks with delayed knowledge of the channel state information have received significant attention in recent years. An exception is networks where channels are spatially correlated. This paper characterizes the capacity region of two-user erasure interference channels with delayed knowledge of the channel state information and spatially correlated channels. There are instances where spatial correlation eliminates any potential gain from delayed channel state information and instances where it enables the same performance that is possible with instantaneous knowledge of channel state. The key is an extremal entropy inequality for spatially correlated channels that separates the two types of instances. It is also shown that to achieve the capacity region, each transmitter only needs to rely on the delayed knowledge of the channels to which it is connected.
\end{abstract}



\section{Introduction}
\label{Section:Introduction}

Interference is the main bottleneck in fast-fading mobile networks where the channels are in constant flux and the nodes may only have access to partial or delayed knowledge of the channel state. Thus, when designing transmission protocols, we ought to make the best use of the available limited knowledge at each node. Researchers have considered wireless networks in which transmitters have access to delayed knowledge of the channel state information (CSI). In particular, the delayed knowledge was used in~\cite{jolfaei1993new} to create transmitted signals that are simultaneously useful for multiple users in a broadcast channel. These ideas were then extended to different wireless networks, including the erasure broadcast channels~\cite{georgiadis2009broadcast}, leading to determination of the DoF region of broadcast channels~\cite{maddah2012completely}, and the DoF region of multi-antenna multi-user Gaussian ICs and X channels~\cite{GhasemiX1,Vaze_DCSIT_MIMO_BC,Jafar_Retrospective,vahid2015informational}.

In most prior work, channel gains are assumed to be independently and identically distributed across {\it time} and {\it space}. However, such assumptions are not realistic, and as we show later, may lead to underestimates of system capacity. There have been some results for the case where channel gains are correlated across time~\cite{yang2013degrees,chen2012degrees}, the idea being that correlation across time allows transmitters to better estimate what will happen next, and to adjust their transmission strategies accordingly. By contrast, spatial correlation has somewhat been neglected.  

Here, we study how spatial correlation of channel gains affects the capacity region and the transmission strategies of fast-fading wireless networks where transmitters have access to delayed CSI. We adopt a noiseless channel model that abstracts important facets of communication in the real world. The two-user Erasure IC model was introduced in~\cite{AlirezaBFICDelayed}, to understand the capacity region of the two-user fading ICs under delayed channel state information assumption. It was later shown that this model abstracts important aspects of wireless packet networks~\cite{vahid2014communication}, and networks with varying topologies~\cite{sun2013topological}.  In this model, the channel gains at each time are in the binary field given by Bernoulli $\mathcal{B}\left( p \right)$ distribution.

\begin{figure}[t]
\centering
\includegraphics[height = 4cm]{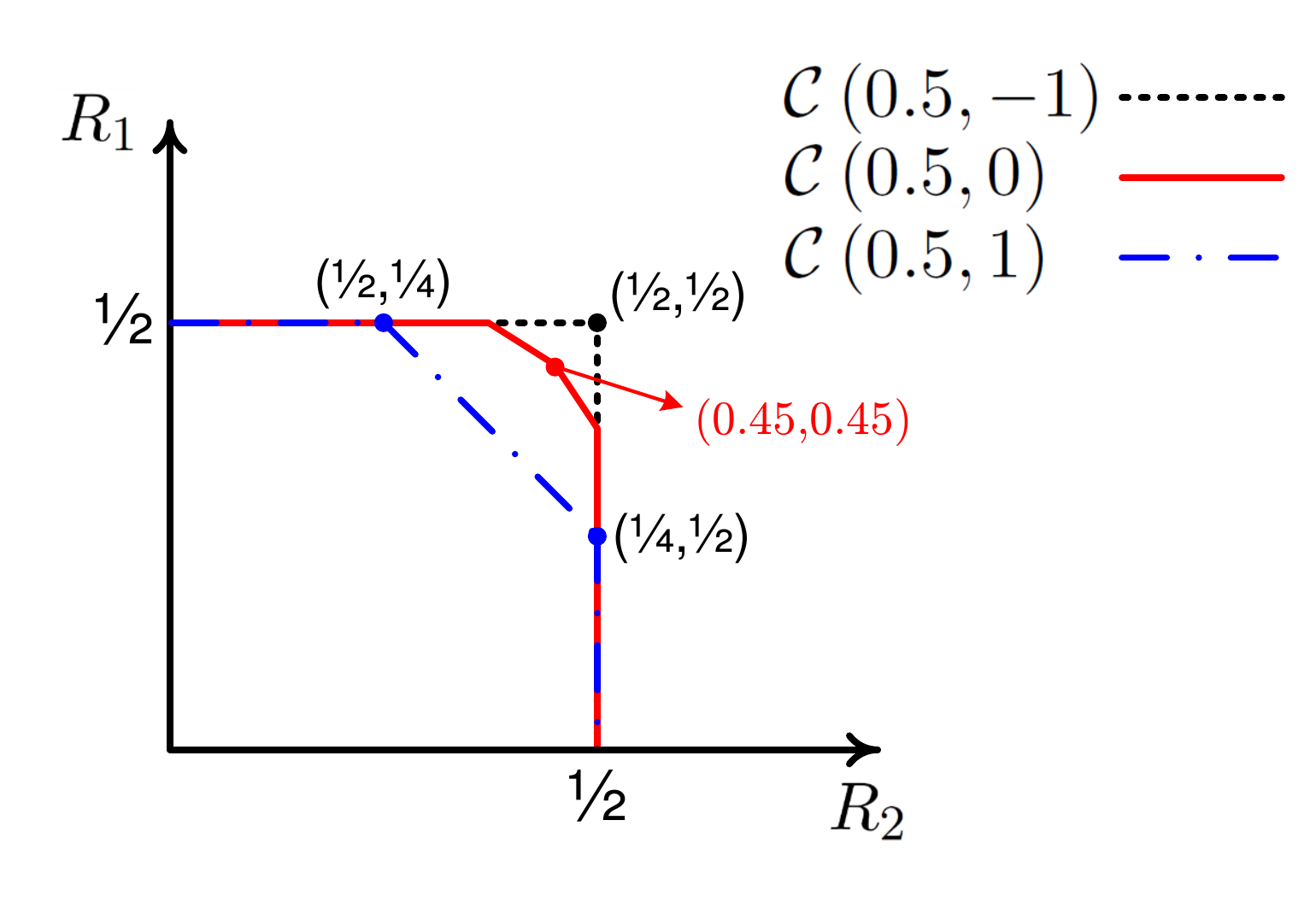}
\caption{\it Capacity region for $p =0.5$ and three different correlation coefficients. Spatial correlation can reduce the region to that of no CSIT, or improve it to that of instantaneous CSIT.\label{Fig:CapacityHalf}}
\vspace{-2.5mm}
\end{figure}

Here, we assume that the transmitters become aware of the channel realizations with unit delay. Moreover, we assume that there is certain spatial correlation between the links connected to each transmitter and that this knowledge is available to all nodes as side information. We show that spatial correlation can greatly affect the capacity region of the two-user erasure IC with delayed channel state information at the transmitters (CSIT). In fact, spatial correlation on the one hand can take away any potential gain of delayed CSIT and on the other hand, it can help us perform as well as having instantaneous knowledge. To see this dichotomy, we focus on $p =0.5$. When the two channel gains connected to each transmitter are fully correlated (\emph{i.e.} correlation coefficient of $1$), the capacity region $\mathcal{C}\left( 0.5, 1\right)$ coincides with the one where transmitters do not have any access to the channel state information (even with delay) as shown in Fig.~\ref{Fig:CapacityHalf}. At the other extreme, when the two channel gains connected to each transmitter have a correlation coefficient of $-1$, the gain of delayed CSI is accentuated and the capacity region $\mathcal{C}\left( 0.5, -1\right)$ matches that where transmitters have access to perfect and instantaneous knowledge of the channel state information. 

To derive the new outer-bounds that capture the spatial correlation of the channels, we develop an extremal entropy inequality. This inequality quantifies the extent to which a  transmitter can favor one receiver over another. We show that this inequality is tight and can be in fact achieved. We then use genie-aided arguments and apply our extremal entropy inequality to derive the outer-bounds. 

We know the value of wireless is simultaneous communication to many nodes, and that the challenge is that everyone interferes with everyone else. This intuition informs our transmission strategy which is divided into three phases. After the initial phase, each transmitter takes advantage of the delayed CSI and the correlation information to create new symbols that are  simultaneously interesting to both receivers. This mitigates the interference issue in future phases of communication and allows each transmitter to achieve the limit set by the extremal entropy inequality. 


\section{Problem Formulation}
\label{Section:Problem}

We would like to understand the impact of spatial correlation on the capacity region of the two-user fading interference channels with delayed CSIT. To do so, we consider a noiseless erasure model introduced in~\cite{AlirezaBFICDelayed}. In the erasure model, the channel gain from transmitter ${\sf Tx}_i$ to receiver ${\sf Rx}_j$ takes values in the binary field, and at time $t$ is denoted by $G_{ij}[t] \in \{0,1\}$, $i,j \in \{1,2\}$. Channel gains are distributed as Bernoulli random variables with parameter $p$, \emph{i.e.} $G_{ij}[t] \overset{d}\sim \mathcal{B}(p)$, $i,j = 1,2$. We set $q \overset{\triangle}=1-p$. We define the channel state information at time instant $t$ to be the set
\begin{align}
G[t] \overset{\triangle}= \left \{ G_{11}[t], G_{12}[t], G_{21}[t], G_{22}[t] \right \}.
\end{align}

\vspace{-4mm}
\begin{figure}[ht]
\centering
\includegraphics[height = 3.5cm]{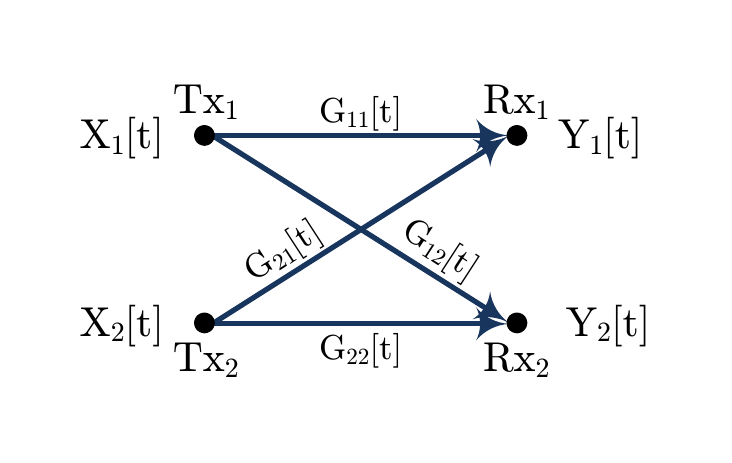}
\caption{Two-user Erasure Interference Channel.\label{Fig:detIC}}
\end{figure}
We assume that transmitter ${\sf Tx}_i$ becomes aware of the CSI with unit delay. Since each receiver only needs to decode its message at the end of the communication block, without loss of generality, we assume that each receiver has instantaneous knowledge of the channel state information.

The input-output relation of this channel at time $t$ is given by
\begin{equation} 
Y_i[t] = G_{ii}[t] X_i[t] \oplus G_{\bar{i}i}[t] X_{\bar{i}}[t], \quad i = 1, 2,
\end{equation}
where $\bar{i} = 3 - i$, $G_{ii}[t], G_{\bar{i}i}[t] \in \{ 0, 1\}$, $X_i[t] \in \{ 0, 1\}$ is the transmit signal of transmitter $i$ at time $t$, and $Y_i[t] \in \{ 0, 1\}$ is the observation of receiver $i$ at time $t$. All algebraic operations are in $\mathbb{F}_2$. 

We further assume that the channel gains are distributed independently across time, and that the channels connected to different transmitters are also independent, \emph{i.e.}
\begin{align}
G_{1j}[t] \perp G_{2\ell}[t], \qquad j,\ell \in \{ 1, 2 \}.
\end{align}
However, we assume that the channel gains corresponding to the links connected to ${\sf Tx}_i$ have a correlation coefficient $\rho$, \emph{i.e.}
\begin{align}
\rho = \frac{\mathrm{cov}\left( G_{i1}[t], G_{i2}[t] \right)}{\sigma_{G_{i1}[t]}\sigma_{G_{i2}[t]}}, \qquad i=1,2. 
\end{align}

We note that fixing $-1 \leq \rho \leq 1$ imposes a feasible set on $p$ denoted by $\mathcal{S}_{\rho}$. More pecisely, if we define
\begin{align}
\label{Eq:Pij}
p_{ij} \overset{\triangle}= \Pr\left( G_{11}[t] = i, G_{12}[t] = j \right), \qquad i,j \in \{ 0, 1 \},
\end{align}
then, $\mathcal{S}_{\rho} \subseteq \left[ 0, 1\right]$ is the set of all values for $p$ such that
\begin{equation}
\label{Eq:Range}
\left\{ \begin{array}{ll}
\vspace{1mm} 0 \leq p_{00}, p_{10}, p_{01}, p_{11} \leq 1, &  \\
\vspace{1mm} p_{10} + p_{11} = p, &  \\
\vspace{1mm} p_{01} + p_{11} = p, &  \\
\vspace{1mm} \sum_{i,j \in \{ 0, 1 \}}p_{ij} = 1, &  \\
p q \rho = p_{00} p^2 + \left( p_{01} + p_{10} \right) p q + p_{11} q^2. &
\end{array} \right.
\end{equation}
It is easy to verify that
\begin{align}
\mathcal{S}_{\rho} \overset{\triangle}= \left[ \max\left\{ 0, \frac{-\rho}{1-\rho} \right\}, \min\left\{ 1, \frac{1}{1-\rho} \right\} \right], 
\end{align}
where we set $\mathcal{S}_{1} \overset{\triangle}= \left[ 0, 1 \right]$. We note that for $\rho = -1$, we get $\mathcal{S}_{-1} = \left\{ 1/2 \right\}$.

We consider the scenario in which ${\sf Tx}_i$ wishes to reliably communicate message $\hbox{W}_i \in \{ 1,2,\ldots,2^{n R_i}\}$ to ${\sf Rx}_i$ during $n$ channel uses, $i = 1,2$. We assume that the messages and the channel gains are {\it mutually} independent and the messages are chosen uniformly. Let message $\hbox{W}_i$ be encoded as $X_i^n$ at transmitter ${\sf Tx}_i$, where at time $t$ we have $X_i[t] = f_i(\hbox{W}_i, G^{t-1})$. Receiver ${\sf Rx}_i$ is only interested in decoding $\hbox{W}_i$, and it will decode the message using the decoding function $\widehat{\hbox{W}}_i = g_i(Y_i^n,G^n)$. An error occurs when $\widehat{\hbox{W}}_i \neq \hbox{W}_i$. The average probability of decoding error is given by
\begin{equation}
\label{eq:errorterms}
\lambda_{i,n} = \mathbb{E}[P[\widehat{\hbox{W}}_i \neq \hbox{W}_i]], \hspace{5mm} i = 1, 2,
\end{equation}
and the expectation is taken with respect to the random choice of the transmitted messages $\hbox{W}_1$ and $\hbox{W}_2$. 

A rate-tuple $(R_1,R_2)$ is said to be achievable if there exist encoding and decoding functions at the transmitters and the receivers respectively, such that the decoding error probabilities $\lambda_{1,n},\lambda_{2,n}$ go to zero as $n$ goes to infinity. The capacity region for $p \in \mathcal{S}_{\rho}$, $\mathcal{C}\left( p, \rho \right)$, is the closure of all achievable rate-tuples. In the following section, we present our main results.


\section{Statement of the Main Results}
\label{Section:Main}

The following theorem describes the capacity region of the two-user erasure interference channel under delayed CSIT assumption where the links connected to the same transmitter have a correlation coefficient of $\rho$.

\begin{theorem}
\label{THM:CapacityCorrelated}
For the two-user erasure interference channel with delayed CSIT and correlated links as described in Section~\ref{Section:Problem} and for $p \in \mathcal{S}_{\rho}$\footnote{For $p=0$, the region is simply the origin.}, we have
\begin{equation}
\label{Eq:Capacity}
\mathcal{C}\left( p, \rho \right) =
\left\{ \begin{array}{ll}
\vspace{1mm} 0 \leq R_i \leq p, &  i = 1,2, \\
R_i + \beta R_{\bar{i}} \leq \beta \left( 1 - q^2 \right), & i = 1,2.
\end{array} \right.
\end{equation}
where 
\begin{align}
\label{Eq:Beta}
\beta = \frac{2p-pq\rho-p^2}{p}.
\end{align}
\end{theorem}

The converse proof of Theorem~\ref{THM:CapacityCorrelated} relies on an extremal entropy inequality for correlated channels that we present in Section~\ref{Section:Converse}. Using this extremal inequality and genie-aided arguments, we obtain the outer-bound. The achievability strategy at transmitter ${\sf Tx}_i$ relies solely on the delayed knowledge of the outgoing links connected to it (\emph{i.e.} at time $t$, $X_i[t]$ is a function of $W_i$, $G_{i1}^{t-1}$ and $G_{i2}^{t-1}$). The transmission strategy can include as many as three phases of communications. After each phase, transmitters use the delayed CSIT and the correlation side information to smartly retransmit equations in a way to help receivers decode their corresponding messages. The key is to combine previously transmitted bits in a way to create as many equations as possible that are of interest to both receivers. 

\begin{figure}[ht]
\centering
\includegraphics[width = \columnwidth]{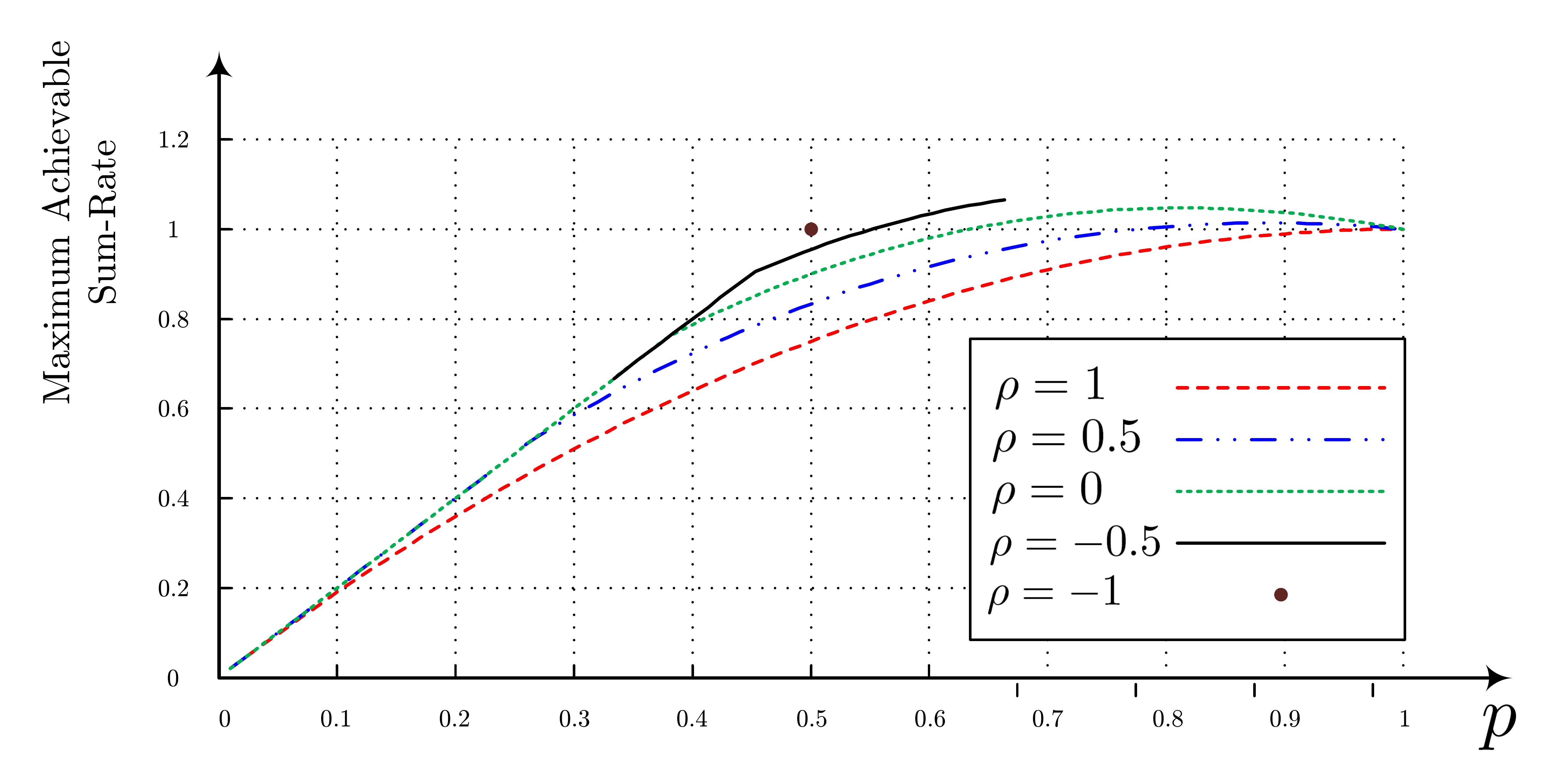}
\caption{\it Maximum achievable sum-rate for $\rho \in \{ -1, -0.5, 0, 0.5, 1 \}$ and $p \in \mathcal{S}_{\rho}$.\label{Fig:SumCap}}
\end{figure}

Before presenting the proofs, we provide further interpretation of Theorem~\ref{THM:CapacityCorrelated}. Fig.~\ref{Fig:SumCap} depicts the maximum achievable sum-rate for $\rho \in \{ -1, -0.5, 0, 0.5, 1 \}$ and $p \in \mathcal{S}_{\rho}$. For a fixed value of $p$ as $\rho$ moves from $+1$ to $-1$, the maximum achievable sum-rate improves. In fact for $p=0.5$ as discussed in the introduction and shown in Fig.~\ref{Fig:CapacityHalf}, $\mathcal{C}\left( 0.5, 1\right)$ coincides with the capacity region of the two-user erasure IC with no CSIT. On the other hand, $\mathcal{C}\left( 0.5, -1\right)$ includes $\left( R_1, R_2 \right) = \left( 0.5, 0.5 \right)$ which implies that the capacity region coincides with that of instantaneous CSIT assumption. Intuitively, this is due to the fact that with fully correlated channels, each transmitter cannot distinguish between the two receivers. However with negative correlation, a transmitter's power to favor (in terms of the received entropy) one receiver over the other improves. When channel gains are distributed independently and identically across time and space, the capacity is given by  $\mathcal{C}\left( 0.5, 0\right)$. These arguments will be made mathematically precise in the following sections.


\section{Converse Proof of Theorem~\ref{THM:CapacityCorrelated}}
\label{Section:Converse}

The derivation of the outer-bounds on the individual rates is straightforward and thus omitted. To derive the other bounds, we first present an extremal entropy lemma tailored to correlated channels with delayed CSIT.

\begin{lemma}
\label{Lemma:Leakage}
For the two-user erasure interference channel with delayed CSIT and correlated links as described in Section~\ref{Section:Problem} and for $p \in \mathcal{S}_{\rho}$ and $p \neq 0$, we have
\begin{align}
H\left( Y_2^n | W_2, G^n \right) \geq \frac{1}{\beta} H\left( Y_1^n | W_2, G^n \right),
\end{align}
where $\beta$ is given in (\ref{Eq:Beta}).
\end{lemma}

\begin{proof} 
For time instant $t$ where $1 \leq t \leq n$, we have
\begin{align}
& H\left( Y_2[t] | Y_2^{t-1}, W_2, G^t \right) \\
& \quad = p H\left( X[t] | Y_2^{t-1}, G_{12}[t] = 1, W_2, G^{t-1} \right) \nonumber \\
& \quad \overset{(a)}= p H\left( X[t] | Y_2^{t-1}, W_2, G^t \right) \nonumber \\
& \quad \overset{(b)}\geq p H\left( X[t] | Y_1^{t-1},Y_2^{t-1}, W_2, G^t \right) \nonumber \\
& \quad = \frac{p}{2p-pq\rho-p^2} H\left( Y_1[t], Y_2[t] | Y_1^{t-1},Y_2^{t-1}, W_2, G^t \right), \nonumber 
\end{align}
where $(a)$ holds since $X[t]$ is independent of the channel realization at time instant $t$; and $(b)$ follows from the fact that conditioning reduces entropy. Therefore, we have 
\begin{align}
& \sum_{t=1}^n{H\left( Y_2[t] | Y_2^{t-1}, W_2, G^t \right)} \\
&~\geq \frac{p}{2p-pq\rho-p^2} \sum_{t=1}^n{H\left( Y_1[t], Y_2[t] | Y_1^{t-1}, Y_2^{t-1}, W_2, G^t \right)}, \nonumber 
\end{align}
and since the transmit signals at time instant $t$ are independent from the channel realizations in future time instants, we get
\begin{align}
& H\left( Y_2^n | W_2, G^n \right) \nonumber \\
& \quad \geq \frac{1}{\beta} H\left( Y_1^n, Y_2^n | W_2, G^n \right) \geq \frac{1}{\beta} H\left( Y_1^n | W_2, G^n \right).
\end{align}
\end{proof}

Using Lemma~\ref{Lemma:Leakage}, we have
\begin{align}
n &\left( R_1 + \beta R_2 \right) = H(W_1) + \beta H(W_2) \nonumber \\
& = H(W_1|W_2, G^n) + \beta H(W_2|G^n) \nonumber \\
& \overset{(\mathrm{Fano})}\leq I(W_1;Y_1^n|W_2,G^n) + \beta I(W_2;Y_2^n|G^n) + n \epsilon_n \nonumber \\
& = H(Y_1^n|W_2,G^n) - \underbrace{H(Y_1^n|W_1,W_2,G^n)}_{=~0} \nonumber \\
& \quad + \beta H(Y_2^n|G^n) - \beta H(Y_2^n|W_2,G^n) + n \epsilon_n \nonumber \\
& \overset{\textrm{Lemma}~\ref{Lemma:Leakage}}\leq \beta H(Y_2^n|G^n) + n \epsilon_n \nonumber \\
& \leq n \beta (1-q^2) + \epsilon_n.
\end{align}
Dividing both sides by $n$ and let $n \rightarrow \infty$, we get
\begin{align}
R_1 + \beta R_2 \leq \beta (1-q^2).
\end{align}


\section{Achievability Proof of Theorem~\ref{THM:CapacityCorrelated}}
\label{Section:Achievability}

In this section for $p \in \mathcal{S}_{\rho}$, we provide the achievability strategy for the maximum symmetric sum-rate point as given by
\begin{align}
\label{Eq:MaxSumRate}
R_1 = R_2 = \min\left\{ p, \frac{\beta\left( 1 - q^2 \right)}{1+\beta} \right\},
\end{align}
where ignoring the degenerate case of $p = 0$, 
\begin{align}
\beta = \frac{2p-pq\rho-p^2}{p}.
\end{align}
The achievability strategy for other corner points of the capacity region follows similar principles with special consideration to the asymmetric rates of ${\sf Tx}_1$ and ${\sf Tx}_2$.

Suppose each transmitter wishes to communicate $m$ bits to its intended receiver. It suffices to show that this task can be accomplished (with vanishing error probability as $m \rightarrow \infty$) in
\begin{align}
\max\left\{ \frac{1}{p}, \frac{1+\beta}{\beta\left( 1 - q^2 \right)} \right\} m + \mathcal{O}\left( m^{\frac{2}{3}} \right)
\end{align}
time instants. The transmission protocol that achieves capacity is divided into the three phases described below. We note that our transmission strategy here follows the general ideas of our earlier result on local delayed CSIT~\cite{vahid2015value}. However due to spatial correlation assumption, new challenges arise and the design parameters need to be carefully chosen to achieve the outer-bounds.

\begin{figure}[ht]
\centering
\includegraphics[width = 8cm]{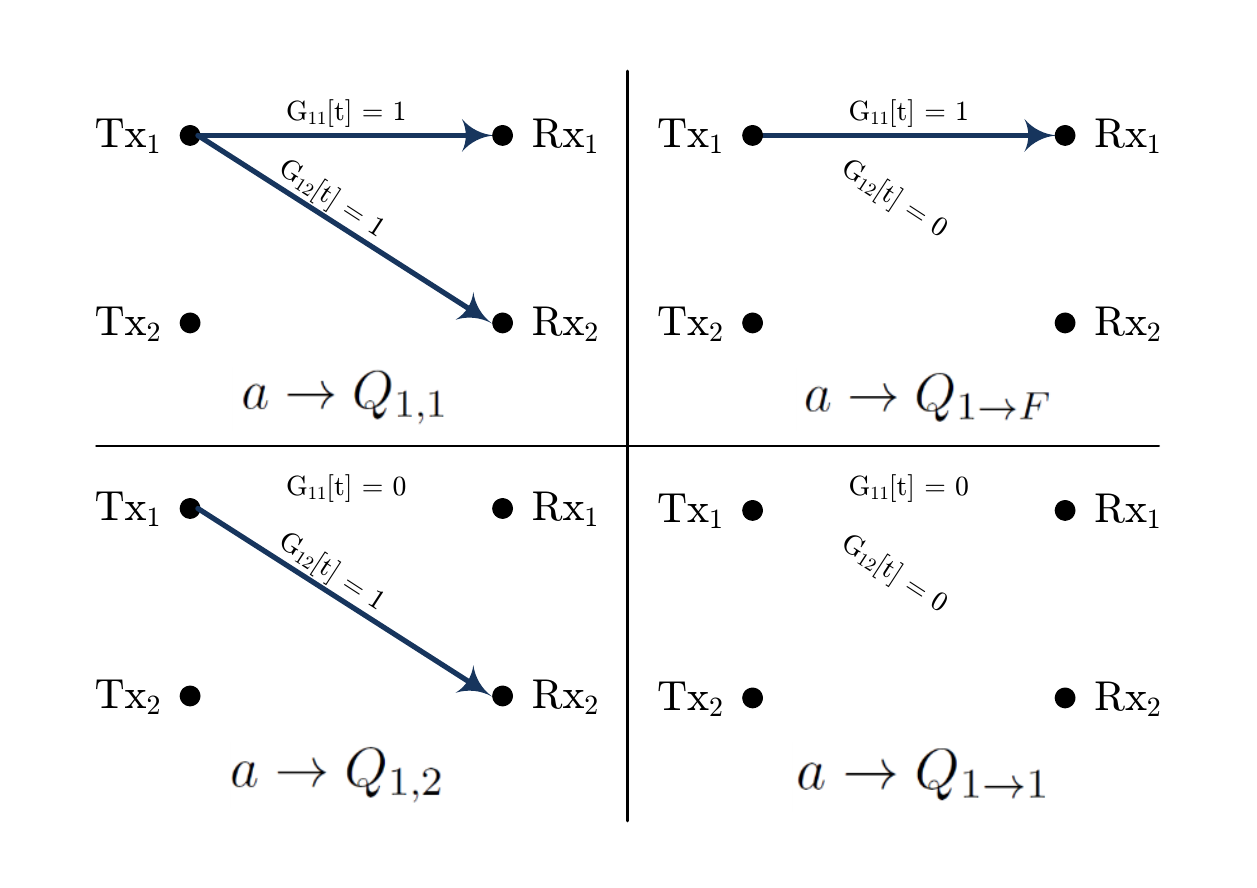}
\caption{\it Based on the values of $G_{11}[t-1]$ and $G_{12}[t-1]$, the status of packet $a$ is updated.\label{Fig:TableFig}}
\end{figure}

\noindent {\bf Phase 1}: At the beginning of the communication block, we assume that the $m$ bits at ${\sf Tx}_i$ are in a queue\footnote{We assume that the queues are column vectors and bits are placed in each queue according to the order they join in.} denoted by $Q_{i \rightarrow i}$, $i=1,2$. At each time instant $t$, ${\sf Tx}_i$ transmits a bit from $Q_{i \rightarrow i}$, and this bit will either stay in this initial queue or will transition to one of the queues listed in Fig.~\ref{Fig:TableFig}. If at time instant $t$, $Q_{i \rightarrow i}$ is empty, then ${\sf Tx}_i$, $i=1,2$, remains silent until the end of Phase 1.
\begin{enumerate}
\item [(A)] $Q_{i \rightarrow F}$: The bits for which no retransmission is required and thus we consider delivered;

\item [(B)] $Q_{i, 1}$: The bits for which at the time of communication, all channel gains known to ${\sf Tx}_i$ with unit delay were equal to $1$;

\item [(C)] $Q_{i, 2}$: The bits for which at the time of communication, we have $G_{ii}[t] = 0$ and $G_{i\bar{i}}[t] = 1$.
\end{enumerate}

For $p \in \mathcal{S}_{\rho}$, it is straightforward to see that
\begin{align}
p_{11} & = pq\rho+p^2, \nonumber \\
p_{10} & = p - pq\rho - p^2, \nonumber \\
p_{01} & = p - pq\rho - p^2, \nonumber \\
p_{00} & = 1 - p_{11} - p_{10} - p_{01}.
\end{align}

Phase~$1$ continues for 
\begin{align}
\frac{1}{p_{00}} m + m^{\frac{2}{3}}
\end{align}
time instants, and if at the end of this phase, either of the queues $Q_{i \rightarrow i}$ is not empty, we declare error type-I and halt the transmission (we assume $m$ is chosen such that $m^{\frac{2}{3}} \in \mathbb{Z}$). 

Assuming that the transmission is not halted, let $N_{i,1}$ and $N_{i, 2}$ denote the number of bits in queues $Q_{i,1}$ and $Q_{i,2}$ respectively at the end of the first phase, $i=1,2$. The transmission strategy will be halted and error type-II occurs, if any of the following events happens.
\begin{align}
\label{eq:errortypeII}
& N_{i,1} > \mathbb{E}[N_{i,1}] + 2 m^{\frac{2}{3}} \overset{\triangle}= n_{i,1}, \quad i=1,2; \nonumber \\
& N_{i,2} > \mathbb{E}[N_{i,2}] + 2 m^{\frac{2}{3}} \overset{\triangle}= n_{i,2}, \quad i=1,2.
\end{align}

From basic probability, we have
\begin{align}
\label{eq:expectedvalues}
\mathbb{E}[N_{i,1}] = \frac{p_{11}m}{p_{11}+p_{10}+p_{01}},~\mathbb{E}[N_{i,2}] = \frac{p_{01}m}{p_{11}+p_{10}+p_{01}}.
\end{align}
At the end of Phase $1$, we add $0$'s (if necessary) in order to make queues $Q_{i,1}$ and $Q_{i,2}$ of size equal to $n_{i,1}$ and $n_{i,2}$ respectively as given above, $i=1,2$.

Moreover since channel gains connected to each receiver are distributed independently, statistically a fraction of $q$ of the bits in $Q_{i,1}$ and a fraction of $q$ of the bits in $Q_{i,2}$ are known to ${\sf Rx}_{\bar{i}}$, $i=1,2$. Denote the number of bits in $Q_{i,j}$ known to ${\sf Rx}_{\bar{i}}$ by
\begin{align}
N_{i,j|{\sf Rx}_{\bar{i}}}, \qquad i,j \in \{ 1, 2 \}.
\end{align}
At the end of communication, if we have
\begin{align}
\label{Eq:KnownBits}
N_{i,j|{\sf Rx}_{\bar{i}}} < q n_{i,j} - m^{\frac{2}{3}}, \qquad i,j \in \{ 1, 2 \},
\end{align}
we declare error type-III. Note that transmitters cannot detect error type-III, but receivers have sufficient information to do so. 

Furthermore using the Bernstein inequality, we can show that the probability of errors of types I, II, and III decreases exponentially with $m$. For the rest of this subsection, we assume that Phase~1 is completed and no error has occurred.

Transmitter ${\sf Tx}_i$ creates two matrices $\mathbf{C}_{i,1}$ and $\mathbf{C}_{i,2}$, $i=1,2$, of sizes $\left( n_{i,1} + 2 m^{\frac{2}{3}} \right) \times n_{i,1}$ and $\left( n_{i,2} + 2 m^{\frac{2}{3}} \right) \times n_{i,2}$ respectively, where entries to each matrix are drawn from i.i.d. $\mathcal{B}(0.5)$ distribution. We assume that these matrices are generated prior to communication and are shared with receivers. Transmitter ${\sf Tx}_i$ does not need to know $\mathbf{C}_{\bar{i},1}$ or $\mathbf{C}_{\bar{i},2}$, $i=1,2$. Note that as $m \rightarrow \infty$, these matrices have full column-rank with probability $1$. We refer the reader for a detailed discussion on the rank of randomly generated matrices in a finite field to~\cite{bourgain2010singularity}.

\noindent {\bf Phase 2}: This phase depends on the values of $p$ and $\rho$. If $n_{i,2} \leq n_{i,1}$, transmitter ${\sf Tx}_i$ combines the bits in $Q_{i,1}$ and $Q_{i,2}$ to create $\tilde{Q}_i$ using the following equation.
\begin{align}
\tilde{Q}_i \overset{\triangle}= \left[ \mathbf{C}_{i,1} Q_{i,1} \right]_{\text{rows~}1:n_{i,2}+m^{\frac{2}{3}}} \oplus \mathbf{C}_{i,2} Q_{i,2}, \quad i=1,2.
\end{align}
However if $n_{i,2} > n_{i,1}$, transmitter ${\sf Tx}_i$ combines the bits in $Q_{i,1}$ and $Q_{i,2}$ to create $\tilde{Q}_i$ using the following equation.
\begin{align}
\tilde{Q}_i \overset{\triangle}= \mathbf{C}_{i,1} Q_{i,1} \oplus \left[ \mathbf{C}_{i,2} Q_{i,2} \right]_{\text{rows~}1:n_{i,1}+m^{\frac{2}{3}}}, \quad i=1,2.
\end{align}

Then the goal is to provide the bits in $\tilde{Q}_1$ and $\tilde{Q}_2$ to \emph{both} receivers. The problem resembles a network with two transmitters and two receivers where each transmitter ${\sf Tx}_i$ wishes to communicate an independent message $\hbox{W}_i$ to {\it both} receivers, $i=1,2$. The channel gain model is the same as described in Section~\ref{Section:Problem}. We refer to this problem as the  two-multicast problem. It is a straightforward exercise to show that for this problem, a rate-tuple of $\left( R_1, R_2 \right) = \left( \frac{\left( 1- q^2 \right)}{2}, \frac{\left( 1- q^2 \right)}{2} \right)$ is achievable.


Transmitters encode and communicate the bits in $\tilde{Q}_1$ and $\tilde{Q}_2$ using the achievability strategy of the two-multicast problem during Phase~2. This phase lasts for 
\begin{align}
2 \left( 1 - q^2 \right)^{-1} \min \left\{ n_{i,1}, n_{i,2} \right\}, \qquad i = 1,2,
\end{align} 
time instants. We assume $\tilde{Q}_1$ and $\tilde{Q}_2$ are decoded successfully at both receivers and no error has occurred. We note that if $n_{i,2} = n_{i,1}$, the transmission strategy ends here.

\noindent {\bf Phase 3}: In this phase, the goal is to finish delivering the remaining bits from Phase~2. If $n_{i,2} < n_{i,1}$, transmitter ${\sf Tx_i}$ encodes and communicates
\begin{align}
\left[ \mathbf{C}_{i,1} Q_{i,1} \right]_{\text{rows~}n_{i,2}+m^{\frac{2}{3}}+1:n_{i,1}+2m^{\frac{2}{3}}}
\end{align}
using the achievability strategy of the two-multicast problem. On the other hand, if $n_{i,2} > n_{i,1}$, transmitter ${\sf Tx_i}$ encodes and communicates
\begin{align}
\left[ \mathbf{C}_{i,2} Q_{i,2} \right]_{\text{rows~}n_{i,1}+m^{\frac{2}{3}}+1:n_{i,2}+2m^{\frac{2}{3}}}
\end{align}
using a point-to-point erasure code of rate $p$.

\noindent {\bf Decoding}: Due to the statistics of the channel, some of the bits in $Q_{\bar{i},1}$ and $Q_{\bar{i},2}$ are already known to ${\sf Rx}_i$ (the number of known bits is given in (\ref{Eq:KnownBits})). Receiver ${\sf Rx}_i$ removes these known bits from its received signal. Then upon successful completion of Phases~2 and~3, receiver ${\sf Rx}_i$ obtains $p \left( n_{\bar{i},1} + n_{\bar{i},2} \right)$ linearly independent equations of at most $p \left( n_{\bar{i},1} + n_{\bar{i},2} \right)$ remaining bits in $Q_{\bar{i},1}$ and $Q_{\bar{i},2}$, and with probability $1$ as $m \rightarrow \infty$ it is able to cancel out the interference from ${\sf Tx}_{\bar{i}}$. After removing the interference, ${\sf Rx}_i$ has access to $m$ random linear combinations of its intended $m$ bits (with probability $1$ as $m \rightarrow \infty$). As a result, the bits intended for ${\sf Rx}_i$ can be reconstructed from the available linear combinations.

The total communication time is then equal to the length of Phases~1,~2, and~3. Thus asymptotically, the total communication time is
\begin{align}
\max\left\{ \frac{1}{p}, \frac{1+\beta}{\beta\left( 1 - q^2 \right)} \right\} m + \mathcal{O}\left( m^{\frac{2}{3}} \right)
\end{align}
time instants and that completes the achievability proof of the maximum symmetric sum-rate as given in (\ref{Eq:MaxSumRate}).


%


%


\section{Conclusion and Future Directions}
\label{Section:Conclusion}

We studied the impact of spatial correlation of channel gains on the capacity region and on the transmission strategies of fast-fading wireless networks with delayed CSIT. We quantified this impact by deriving a new extremal entropy inequality. To achieve the capacity region, we developed a novel transmission protocol, incorporating as many as three phases of communication, that takes full advantage of delayed CSI and side information about spatial correlation. A natural next step is to consider the situation where the spatial correlation is present both at the transmitter side and at the receiver side.   

%
%
%

\bibliographystyle{ieeetr}
\bibliography{bib_FBBudget}

\end{document}